\documentclass[review,hidelinks,onefignum,onetabnum]{siamart220329}

\usepackage{lipsum}
\usepackage{mathtools}
\usepackage{amsfonts}
\usepackage{graphicx}
\usepackage{epstopdf}
\usepackage{algorithmic}
\usepackage[caption=false]{subfig}
\usepackage{pgfplots}

\usepackage{caption}
\ifpdf
  \DeclareGraphicsExtensions{.eps,.pdf,.png,.jpg}
\else
  \DeclareGraphicsExtensions{.eps}
\fi

\newsiamremark{remark}{Remark}
\newsiamremark{hypothesis}{Hypothesis}
\crefname{hypothesis}{Hypothesis}{Hypotheses}
\newsiamthm{claim}{Claim}

\title{Parallel square-root statistical linear regression for inference in nonlinear state space models\thanks{Submitted to the editors \funding{This work was funded by Academy of Finland (project 321900)}}}

\author{Fatemeh Yaghoobi, Adrien Corenflos, Sakira Hassan, Simo Särkkä\footnotemark[0] \thanks{F. Yaghoobi, A. Corenflos, S. Hassan, and S. Särkkä are with the Department of Electrical Engineering
and Automation, Aalto University, 02150 Espoo, Finland (email:
fatemeh.yaghoobi@aalto.fi)}}

\usepackage{amsopn}

\ifpdf
\hypersetup{
  pdftitle={Parallel square-root statistical linear regression for inference in nonlinear state space models},
  pdfauthor={F. Yaghoobi, A. Corenflos, S. Hassan, and S. Särkkä}
}
\fi

\begin{document}

\maketitle

\begin{abstract}
    In this article, we first derive parallel square-root methods for state estimation in linear state-space models. We then extend the formulations to general nonlinear, non-Gaussian state-space models using statistical linear regression and iterated statistical posterior linearization paradigms. We finally leverage the fixed-point structure of our methods to derive parallel square-root likelihood-based parameter estimation methods. We demonstrate the practical performance of the methods by comparing the parallel and the sequential approaches on a set of numerical experiments.
\end{abstract}



\section{Introduction}
    Inference in linear and nonlinear state-space models (SSMs) is an active research topic with applications in many fields including target tracking, control engineering, and biomedicine~\cite{bar2004estimation, sarkka2013bayesian}. In this article, we are primarily interested in state and parameter estimation problems in state-space models of the form
    \begin{equation}\label{eq:ss-model}
        \begin{split}
                x_k \mid x_{k-1} &\sim p(x_k \mid x_{k-1}),\\
                y_k \mid x_k &\sim p(y_k \mid x_k),
        \end{split}
    \end{equation}%
    where $k=1,2,\ldots,n$. Above, $x_k \in \mathbb{R}^{n_x}$ and $y_k \in \mathbb{R}^{n_y}$ are the state and measurement at time step $k$, $p(x_k \mid x_{k-1})$ is the transition density of the states, and $p(y_k \mid x_k)$ is the conditional density of the measurements. The initial distribution at time $k=0$ is given by $x_0 \sim p(x_0)$. Our primary focus is on solving the smoothing problem, which consists in estimating the posterior distribution of the state at time step $k$ given a set of observations, that is, $p(x_k \mid y_{1:n}$), for $k=0,1,\ldots,n$. Furthermore, when the transition, measurement, and prior densities of the SSM depend on a vector of unknown parameters $\theta \in \mathbb{R}^{n_\theta}$, we wish to estimate it based on the observations.

    A solution to the smoothing problem is to compute the marginal posterior distribution of the states recursively by using sequential Bayesian filtering and smoothing equations (see, e.g.,~\cite{sarkka2013bayesian}). In this case, when the SSM at hand is linear and Gaussian, the classical Kalman filter~\cite{kalman1960new} and Rauch-Tung-Striebel (RTS) smoother~\cite{rauch1965maximum} provide closed-form solutions in terms of smoothing means and covariances. However, when the SSM is non-Gaussian, closed-form solutions are often not feasible. In such cases, Gaussian approximated solutions can be obtained by using, for example, Taylor expansions or sigma-point methods, with the extended~\cite{jazwinski2007stochastic, maybeck1982stochastic, bar2004estimation} and the unscented~\cite{haykin2001kalman,julier2000new, julier2004unscented, sarkka2008unscented} Kalman filters and smoothers being the most famous instances of these two methods. The accuracy of these methods can be improved by using iterated smoothing approaches such as iterated extended Kalman smoothers (IEKS)~\cite{bell1994iterated} and iterated posterior linearization (IPLS)~\cite{garcia2016iterated}.
    
    IEKS computes the maximum a posteriori (MAP) estimate of the state trajectory through analytical linearization of the SSM, and it can be viewed as an instance of the Gauss--Newton method~\cite{bell1994iterated}. As for standard Gauss--Newton methods, its convergence can be improved by using line search or Levenberg--Marquardt extensions~\cite{sarkka2020levenberg}. IPLS, on the other hand, iteratively applies sigma-point-based transforms to the posterior distribution of the states. The method has been proposed for nonlinear SSMs with additive Gaussian noise in~\cite{garcia2016iterated}, as well as for the general form of SSM~\eqref{eq:ss-model} in~\cite{tronarp2018iterative}. 

    To improve the numerical stability and robustness of filters and smoothers, square-root Kalman filtering and smoothing methods have been introduced in~\cite{Bierman:1977,grewal2001kalman,arasaratnam2008square,arasaratnam2011cubature}. In these methods, the computations involving covariance matrices are reformulated in terms of their Cholesky factors or similar matrix square roots. This ensures that the covariance matrices are guaranteed to stay positive semi-definite and the number of bits required for a given numerical accuracy is roughly half compared to the conventional covariance matrix formulation. Classically, square-root methods have mainly been needed in short word-length architectures arising in systems with strict energy constraints. Interestingly, shorter word lengths are nowadays also used in more modern high-performance computing devices~\cite{jouppi2017datacenter}, making the need for square-root methods more widely prominent. 

    The aforementioned smoothing methods are based on sequential applications of the Bayesian filtering and smoothing equations which have to be performed in a specific order and, therefore, cannot be trivially parallelized. This prevents leveraging many of the substantial recent improvements made in parallel computing hardware architectures such as graphics processing units (GPUs) and tensor processing units (TPUs)~\cite{jouppi2017datacenter}. In recent studies~\cite{sarkka2020temporal, yaghoobi2021parallel, hassan2021temporal}, a solution to this problem was proposed by reformulating the Bayesian filtering and smoothing equations in terms of associative operators, allowing for parallel-in-time processing by using the associative scan algorithm~\cite{blelloch1989scans}. These articles provided parallel methods for linear-Gaussian, nonlinear with additive Gaussian noise, and finite SSMs. In this article, we extend this class of methods significantly by formulating a square-root form of the parallel associative filtering and smoothing operators, leveraging the generalized statistical linear regression and iterated posterior linearization methodologies to tackle general nonlinear state-space models. We combine both of these to provide a general square-root algorithm for nonlinear state-space models. 
 
    The contributions of this article are the following. (i) We develop novel parallel-in-time square-root versions of linear parallel filtering and smoothing methods proposed in~\cite{sarkka2020temporal}. (ii) We extend the iterated parallel filtering and smoothing methods for the additive-noise SSMs~\cite{yaghoobi2021parallel} to the generalized statistical linear regression (GSLR)~\cite{tronarp2018iterative} method, which makes them applicable to general SSMs of the form~\eqref{eq:ss-model}. This is further combined with our novel square-root parallel associative operators to derive parallel-in-time square-root algorithms for fully nonlinear SSMs. (iii) We propose a compute and memory-efficient way to perform parallel-in-time parameter estimation in state-space models by leveraging the fixed-point structure of our algorithms. (iv) The performance of the proposed methods are then empirically validated on a series of realistic examples run on a GPU.

    The article is organized as follows. Section~\ref{sec:background} gives a brief overview of parallel Bayesian smoothers, the square-root filtering and smoothing methods, and the general linearization method for SSMs. Section~\ref{sec:parallel-sqrt-filter-smoother} focuses on developing a square-root extension of the parallel Kalman filter and RTS smoother. In Section~\ref{sec:parallel-filter-smoother-gslr}, we introduce a parallel-in-time iterated filter and smoother based on GSLR linearization for general SSMs in covariance and square-root of covariance forms. In Section~\ref{sec:sys-identification}, we show how the fixed-point structure of the proposed methods can be leveraged to efficiently compute the gradient of the marginal log-likelihood for the SSM at hand. Finally, in Section~\ref{sec:experiment}, we experimentally validate the computational and statistical behavior of the proposed methods. 
    
\section{Background}
\label{sec:background}
    In this section, we review the building blocks of our proposed method. We first discuss the parallel framework of Bayesian filtering and smoothing~\cite{sarkka2020temporal}, then the general square-root formulations arising in the sequential case, and finally, the GSLR method~\cite{tronarp2018iterative}.
    \subsection{Parallel-in-time Bayesian filtering and smoothing} \label{subsec:background-scan}
        Parallel Bayesian filtering and smoothing~\cite{sarkka2020temporal} is a recent framework that enables efficient inference on parallel platforms. It is based on the parallel-scan algorithm~\cite{blelloch1989scans}, which can operate on a given set of elements $\{a_k\}_{1 \leq k \leq n}$, and a binary associative operator $\otimes$. By using this algorithm, it is possible to parallelize the computation of the full prefix-sum of the set, $(a_1 \otimes a_2 \otimes \ldots \otimes a_K)_{1 \leq K \leq n}$, reducing its sequential time complexity from $O(n)$ to a $O(\log n)$ parallel span-complexity. This framework was recently used in~\cite{sarkka2020temporal} to provide a parallel formulation of Bayesian filtering and smoothing. In their approach, they define a set of elements $a_k$ and a binary associative operator $\otimes$ in such a way that the result of $a_1\otimes \dots \otimes a_k$ recovers the filtering distribution $p(x_k \mid y_{1:k})$ at step $k$; they similarly define another set of elements and operator for the smoothing distribution $p(x_k \mid y_{1:n})$. When specialized to linear Gaussian state-space models these equations give the parallel Kalman filter and RTS smoother \cite{sarkka2020temporal}.
        
    \subsection{General triangular square-root method} \label{subec:background-tria}
        A square-root method refers to a method where the matrix square roots of covariances are carried in computations instead of plain covariances \cite{Bierman:1977,grewal2001kalman}. This reduces the number of bits required to store the values while maintaining a large level of accuracy. In this article, we use square-root methods based on triangular (Cholesky) square roots of covariance matrices. Similar to~\cite{arasaratnam2008square}, in order to propagate the square-root form of a covariance matrix $A$ through the corresponding equations, we first leverage the QR-decomposition as $A^\top = Q \, R$, and then we define the triangularization operator as $\mathrm{Tria}(A) = R^\top$. The triangularization operator has the following properties:
        \begin{enumerate}
        \item For any matrix $A \in\mathbb{R}^{n \times m}$, $\mathrm{Tria}(A)$ is a lower triangular matrix of size ${n \times n}$.
        \item For any matrix $A \in\mathbb{R}^{n \times m}$, we have $\mathrm{Tria}(A) \mathrm{Tria}(A)^\top = A A^\top$. \label{prop:tria-prop-2}
        \end{enumerate}
        In particular, the second property implies that, for any block matrix $C = \begin{pmatrix} A & B \end{pmatrix} \in \mathbb{R}^{n \times (m+k)}$, with $A \in \mathbb{R}^{n \times m}$ and $B \in \mathbb{R}^{n \times k}$, we have
        \begin{equation} \label{eq:CCT}
              \mathrm{Tria}(C) \mathrm{Tria}(C)^\top = A A^\top + B B^\top,
          \end{equation} 
        so that we can sum covariance matrices while retaining the square-root structure. Unfortunately, the same method cannot be employed to compute the square root of the expressions of the form $AA^\top - BB^\top$. However, provided that $AA^\top - BB^\top$ is positive definite, we can compute its square root by iteratively applying rank-1 downdate~\cite{gill1974methods,Bierman:1977,krause2015more} to $A$, iterating over the columns of $B$. We denote this operation  by $\mathrm{DownDate}(A,B)$. 
        
    \subsection{Generalized statistical linear regression (GSLR)} \label{subsec:background-gslr}
        GSLR \cite{tronarp2018iterative} is a way to model the relationship between a dependent variable and one or more independent variables and to use this model to perform estimation. Consider a conditional distribution $y \mid x \sim p(y \mid x)$; this type of relationship is ubiquitous in state estimation problems, where the measurement is a function of the state and some noise $e$: $y = g(x, e)$, but it also covers more general cases. The GSLR method~\cite{tronarp2018iterative} approximates this general nonlinear non-Gaussian system as an affine transformation of a Gaussian random variable. In particular, given an arbitrary probability distribution of $x$ as $p(x)$ with moments (i.e., the mean and covariance) $\mathbb{E}[x]$ and $\mathbb{V}[x]$, two conditional moments of $y$, $\mathbb{E}[y \mid x]$ and $\mathbb{V}[y \mid x]$, and the cross-covariance matrix of $x$ and $y$ as $\mathbb{C}[y, x]$, we can find the best affine approximation $y \approx Hx + d + r$ in terms of mean square error as follows:
        \begin{equation}\label{eq:GSLR-linear-params}
             \begin{split}
                H  &= \mathbb{C}[x, y]^\top \mathbb{V}[x]^{-1}, \quad d = \mathbb{E}[y] - H \mathbb{E}[x], \quad \Omega = \mathbb{V}[y] - H \mathbb{V}[x] H^\top.
             \end{split}
        \end{equation}
        The moments involved in the calculation need to be computed as follows:
        \begin{equation}\label{eq:GSLR-moments}%
            \begin{split}
                \mathbb{E}[y] &=\mathbb{E}[\mathbb{E}[y\mid x]], \quad \mathbb{V}[y] = \mathbb{E}[\mathbb{V}[y\mid x]]+\mathbb{V}[\mathbb{E}[y\mid x]],\quad \mathbb{C}[y, x] =  \mathbb{C} [\mathbb{E}[y \mid x], x].
            \end{split}
        \end{equation}
        Since the integrals involved in computing the quantities in~\eqref{eq:GSLR-moments} are not tractable in the general case, we have to resort to approximations. In this paper, we use two well-known strategies, namely Taylor series expansions \cite{jazwinski2007stochastic, maybeck1982stochastic, bar2004estimation} and sigma-point methods \cite{haykin2001kalman,julier2000new, julier2004unscented, sarkka2008unscented}.

\section{Parallel square-root filter and smoother for linear systems}
\label{sec:parallel-sqrt-filter-smoother}
    In the case of linear Gaussian SSMs, S{\"a}rkk{\"a} and  Garc{\'i}a-Fern\'andez in~\cite{sarkka2020temporal} showed how to define $a_k$ and an associative operator $\otimes$ that enable the use of associative scan primitives to compute the Kalman filtering and RTS smoothing marginals in logarithmic time. Here, we will refer to their method as the standard version. In this section, we aim to derive the square-root version of the parallel Kalman filter and RTS smoother. 
    
    Consider the following affine SSM:
    \begin{equation}\label{eq:linear-model}
        \begin{split} 
              x_k &= F_{k-1} x_{k-1} + c_{k-1} + q_{k-1}, \quad y_k = H_k x_k + d_k + v_k,
        \end{split}
    \end{equation}
    where in the standard version we have $x_0 \sim N(m_0, P_0)$, $q_k \sim N(0, \Lambda_k)$, and $v_k \sim N(0, \Omega_k)$ for $k=1,2,\ldots,n$. In the square-root version, we assume that we know the square-root of covariances with $S_{\Lambda_k}$ and $S_{\Omega_k}$ as the Cholesky forms of  $\Lambda_k$ and $\Omega_k$, that is, $\Lambda_k = S_{\Lambda_k} S_{\Lambda_k}^\top$ and $\Omega_k = S_{\Omega_k} S_{\Omega_k}^\top$. Also, we assume that we know the mean $m_0$ and the square root of the covariance $N_0$ of the initial distribution such that $P_0 = N_0 N_0^\top$. The aim of this section is to develop parallel square-root algorithms to compute, first, the filtering distribution, and second, the smoothing distribution. 

    \subsection{Parallel square-root filter for linear systems}
    \label{subsec:parallel-sqrt-filter}
        The aim of this section is to derive a parallel square-root formulation of Kalman filtering. To do so,  we consider  the parallel Kalman filter proposed in~\cite{sarkka2020temporal} for linear Gaussian SSMs of the form~\eqref{eq:linear-model}, and derive the square-root version of the initialization of parameters $a_k$ and the combination of the parameters corresponding to the associative operator $\otimes$.

        Initialization of the parameters in the filtering step in the standard version can be performed by specifying the element $a_k = \big[p(x_k \mid  y_k,x_{k-1}),  p(y_k \mid x_{k-1})\big]$ as a set $\{A_k, b_k, C_k, \eta_k, J_k\}$~\cite[Lemma 7]{sarkka2020temporal} such that:
            \begin{equation*}
                \begin{split}
                   p(x_k \mid  y_k,x_{k-1}) =N(x_k;A_k x_{k-1} + b_k,C_k),\quad p(y_k \mid x_{k-1})\propto N_I(x_{k-1}; \eta_k,J_k),
                \end{split}
            \end{equation*}
            where $N$ denotes a Gaussian distribution, and $N_I$ the information form of a Gaussian distribution \cite{sarkka2013bayesian}. The aim here is to find the square-root version of these initialization parameters, which we define as a set
            \begin{equation} \label{eq:pi-f}
                \Pi^f_k = \{ A_k, b_k, U_k, \eta_k, Z_k \},
            \end{equation}
            where $C_k=U_k U_k^\top$, and $J_k=Z_k Z_k^\top$, respectively. Given these factorizations, we can then obtain $\Pi^f_k$ using the following lemma. 
    
            \begin{lemma}[Initialization step of square-root parallel Kalman filter] \label{lem:sqrt_filter_k1}
            Tables~\ref{tab:sqrt_filter_k1} and~\ref{tab:sqrt_filter_k} present square-root formulations of~\cite{sarkka2020temporal} to find the initialization parameters of parallel Kalman filter, where we define the following matrices using $\mathrm{Tria}$:
            \begin{equation} \label{eq:sqrt-filter-step1}
            \begin{split}
            \begin{pmatrix}
                  \Psi^-_{11} & 0 \\
                  \Psi^-_{21} & \Psi^-_{22}
            \end{pmatrix}
               &= \mathrm{Tria}\left(
                  \begin{pmatrix}
                  H_1 N_1^- & S_{\Omega_1}\\
                  N_1^- & 0
                  \end{pmatrix}
                 \right),
            \end{split}
            \end{equation}
            \begin{equation} \label{eq:sqrt-filter-init-2}
            \begin{split}
            \begin{pmatrix}
                  \Psi_{11} & 0 \\
                  \Psi_{21} & \Psi_{22}
            \end{pmatrix}
               &= \mathrm{Tria}\left(
                  \begin{pmatrix}
                  H_k S_{\Lambda_{k-1}} & S_{\Omega_{k}} \\
                  S_{\Lambda_{k-1}} & 0
                  \end{pmatrix}
                 \right).
            \end{split}
            \end{equation}
            \end{lemma}
            
            \begin{table}[tbhp]
            \begin{center}
              \begin{tabular}{c|c} 
               \textbf{Standard version}~\cite{sarkka2020temporal}& \bf Square-root version  \\ [0.1cm]
               $m_1^- = F_{0} m_{0} + c_{0}$, & $m_1^- = F_{0} m_{0} + c_{0}$,\\ [0.2cm]
                $P_1^- = F_{0} P_{0} F^\top_{0} + \Lambda_{0}$, &  $N_1^- = \mathrm{Tria}\big(\begin{pmatrix}
                F_0 N_0 & S_{\Lambda_0} \end{pmatrix}\big)$,\\ [0.2cm]
                $S_1 = H_1 P_1^- H_1^\top + \Omega_1$, & $Y_1 = \Psi^-_{11}$, \quad \\[0.2cm]
                $K_1 = P_1^- H^\top_1 S_1^{-1}$, & $K_1 = \Psi^-_{21} (\Psi^-_{11})^{-1}$, \\[0.2cm]
                $A_1=0$, & $A_1=0$ , \\[0.2cm]
                $b_1 = m_1^- + K_1 [y_1 - H_1 m_1^- - d_1]$, & $b_1 = m_1^- + K_1 [y_1 - H_1 m_1^- - d_1]$, \\[0.2cm]
                $C_1 = P_1^- - K_1 S_1 K_1^\top$, & $U_1  = \Psi^-_{22}$, \\[0.2cm]
                $J_1 = (H_1 F_{0})^\top S_1^{-1} H_1 F_{0}$, &$Z_1 = F_{0}^{\top}H_{1}^{\top} Y_1^{-\top},$  \\[0.2cm]
                $\eta_1 = (H_1 F_{0} )^\top S_1^{-1} H_1  (y_1 - H_1 c_{0} - d_1)$ & $\eta_1 = Z_1 Y_1^{-1} \left(y_{1}-H_{1}c_{0}-d_{1}\right). $
              \end{tabular}
            \caption{Initialization of the parameters for parallel Kalman filter when $k=1$. On the left there are the covariance-based equations, and on the right there are the proposed square-root versions.}
            \label{tab:sqrt_filter_k1}
            \end{center}
            \end{table}
           \begin{table}[tbhp]
            \begin{center}
            \resizebox{\columnwidth}{!}{%
              \begin{tabular}{c|c} 
               \textbf{ Standard version}~\cite{sarkka2020temporal}& \bf Square-root version  \\ [0.1cm]
                $A_k = (I_{n_x} - K_k H_k)F_{k-1},$ &  $ A_k = (I_{n_x} - K_k H_k)F_{k-1}$,\\ [0.2cm]
                $b_k = c_{k-1} + K_k( y_k - H_k c_{k-1} - d_{k-1})$, & $b_k = c_{k-1} + K_k( y_k - H_k c_{k-1} - d_{k-1})$, \\[0.2cm]
                $C_k = (I_{n_x} -K_k H_k) \Lambda_{k-1}$, & $ U_k = \Psi_{22}$, \\[0.2cm]
                $K_k = \Lambda_{k-1} {H_k}^\top S_k^{-1}$, &$K_k = \Psi_{21} \Psi_{11}^{-1}$  \\[0.2cm]
                $S_k = H_k \Lambda_{k-1} H_k^\top+ \Omega_k,$ & $Y_k = \Psi_{11}$ \\[0.2cm]
                $J_k = (H_k F_{k-1})^\top S_k^{-1} H_k F_{k-1}$ & $Z_k = F_{k-1}^{\top}H_{k}^{\top} Y_k^{-\top}$\footnotemark{}   \\[0.2cm]
                $\eta_k = (H_k F_{k-1} )^\top S_k^{-1} H_k  (y_k - H_k c_{k-1} - d_k)$ & $\eta_k = Z_k Y_k^{-1} \left(y_{k}-H_{k}c_{k-1}-d_{k}\right)$.
              \end{tabular}
            }
            \caption{Initialization of parameters for parallel Kalman filter for $1 < k \leq n$. On the left there are the covariance-based equation, and on the right there are the proposed square-root versions.}
            \label{tab:sqrt_filter_k}
            \end{center}
            \end{table}

            \begin{proof}
            Let $N_1^- =\mathrm{Tria}\big(\begin{pmatrix}
                F_0 N_0 & S_{\Lambda_0} \end{pmatrix}\big)$. 
            By applying~\eqref{eq:CCT}, we obtain $N_1^- [N_1^-]^\top = P_1^-$. Then, by defining the block matrix in \eqref{eq:sqrt-filter-step1} and applying~\eqref{eq:CCT}, we have
            \begin{equation} \label{eq:sqrt-filter-pre-step1}
                \begin{split}
                      \begin{pmatrix}
                      H_1 N_1^- & S_{\Omega_1}\\
                      N_1^- & 0
                      \end{pmatrix}&
                      \begin{pmatrix}
                      H_1 N_1^- & S_{\Omega_1}\\
                      N_1^- & 0
                      \end{pmatrix}^\top = \begin{pmatrix}
                      H_1 P_1^- H_1^T + \Omega_1 & H_1 P_1^- \\
                      P_1^- H_1^\top & 0
                      \end{pmatrix}.
                \end{split}
            \end{equation}
            Noting that $S_1 = Y_1 Y_1^\top$, $C_1=U_1 U_1^\top$, and $J_1=Z_1 Z_1^\top$, the parallel square-root elements $\Pi^f_1$ in \eqref{eq:pi-f} will be obtained as in the equations on the right side of Table~\ref{tab:sqrt_filter_k1}. Please note that $Z_k$ is not a square matrix when $n_x \ne n_y$. In order to make the matrix $Z_k$ square, we can either extend it with zero columns to the right, $Z_k = \begin{pmatrix} F_{k-1}^{\top}H_{k}^{\top} Y_k^{-\top} & 0 \end{pmatrix}$, when $n_x > n_y$ or replace it with $\mathrm{Tria}(Z_k)$ when $n_x < n_y$.

            For $1 < k \leq n$, initialization parameters $A_k$ and $b_k$ are the same in both versions. For the parameters $(Y_k, K_k, U_k)$, by using \eqref{eq:sqrt-filter-init-2}  and applying~\eqref{eq:CCT} we have
            \begin{equation} \label{eq:sqrt-filter-init-1}
            \begin{split}
                \begin{pmatrix}
                  H_k S_{\Lambda_{k-1}} & S_{\Omega_{k}} \\
                  S_{\Lambda_{k-1}} & 0
                  \end{pmatrix} & \begin{pmatrix}
                  H_k S_{\Lambda_{k-1}} & S_{\Omega_{k}} \\
                  S_{\Lambda_{k-1}} & 0
                  \end{pmatrix}^\top = \begin{pmatrix}
                  H_k {\Lambda_{k-1}} H_k^\top & H_k \Lambda_{k-1} \\
                  \Lambda_{k-1} H_k^\top & \Lambda_{k-1}
                  \end{pmatrix},
            \end{split}
            \end{equation}
            which concludes the proof.
            \end{proof}
            Noting that $S_k = Y_k Y_k^\top$, $C_k = U_k U_k^\top$, and $J_k = Z_k Z_k^\top$, the parallel square-root elements $\Pi^f_k$ in \eqref{eq:pi-f} will be obtained as in the equations on the right side of Table~\ref{tab:sqrt_filter_k}. 
            
            The next step consists in describing the binary associative operator $\otimes$. In the standard version~\cite[Lemma 8]{sarkka2020temporal}, the resulting parametric version of the associative filtering operator is given by $\{A_i, b_i, C_i, \eta_i, J_i\} \otimes \{A_j, b_j, C_j, \eta_j, J_k\} \coloneqq \{A_{ij}, b_{ij}, C_{ij}, \eta_{ij}, J_{ij}\}$ which is defined by the equations on the left hand side of Table~\ref{tab:sqrt_combination_filter}. In the square-root version, we replace this operator with 
             \begin{equation}
            \begin{split}
                 \{A_i, b_i, U_i, \eta_i, Z_i\} \otimes \{A_j, b_j, U_j, \eta_j, Z_k\} &=  \{A_{ij}, b_{ij}, U_{ij}, \eta_{ij}, Z_{ij}\} \\
            \end{split}
            \end{equation}
            where $J_i = Z_i Z_i^\top$, $C_i = U_i U_i^\top$, $J_j = Z_j Z_j^\top$, $C_j = U_j U_j^\top$, $J_{ij} = Z_{ij} Z_{ij}^\top$, and $C_{ij} = U_{ij} U_{ij}^\top$.
            \begin{lemma}[Combination step of square-root parallel Kalman filter] \label{lem:sqrt_combination_filter}
            Table~\ref{tab:sqrt_combination_filter} presents a square-root formulation of the combination step of the parallel Kalman filter~\cite{sarkka2020temporal}, where we define the following matrix using $\mathrm{Tria}$:
            \begin{equation} \label{eq:sqrt-filter-com-2}
            \begin{split}
              \begin{pmatrix}
              \Xi_{11} & 0 \\
              \Xi_{21} & \Xi_{22}
              \end{pmatrix}
              =
              \mathrm{Tria}\left(
              \begin{pmatrix}
              U_i^\top Z_j & I_{n_x} \\
              Z_j   &     0
              \end{pmatrix}
              \right).
            \end{split}
            \end{equation}
             \begin{table}[tbhp]
            \begin{center}
            \resizebox{\columnwidth}{!}{%
              \begin{tabular}{c|c} 
               \textbf{Standard version}~\cite{sarkka2020temporal}& \bf Square-root version  \\ [0.1cm]
                $A_{ij} = A_j (I_{n_x} + C_i J_j)^{-1} A_i,$ &  $ A_{ij} = A_j A_i - A_j U_i \Xi_{11}^{-\top} \Xi_{21}^\top A_i$,\\ [0.2cm]
                $b_{ij} = A_j (I_{n_x} + C_i J_j)^{-1} (b_i + C_i \eta_j) + b_j$, & $b_{ij} = A_{j}\left(I_{n_x} - U_i \Xi_{11}^{-\top} \Xi_{21}^\top\right)\left(b_{i}+U_{i} U_{i}^\top\eta_{j}\right)+b_{j}$, \\[0.2cm]
                $C_{ij} = A_j (I_{n_x} + C_i J_j)^{-1} C_i A_j^\top + C_j$, & $U_{ij} = \mathrm{Tria}\left(\begin{pmatrix} A_j U_i \Xi_{11}^{-\top} & U_j \end{pmatrix}\right)$, \\[0.2cm]
                $\eta_{ij} = A_i^\top (I_{n_x} + J_j C_i)^{-1} (\eta_j - J_j b_i) + \eta_i$, &$\eta_{ij} = A_{i}^{\top}\left(I_{n_x} - \Xi_{21} \Xi_{11}^{-1} U_i^\top\right) \left(\eta_{j}-Z_{j}Z_{j}^\top b_{i}\right)+\eta_{i}$,  \\[0.2cm]
                $J_{ij} = A_i^\top (I_{n_x} + J_j C_i)^{-1} J_j A_i + J_i$ & $Z_{ij} = \mathrm{Tria}\left(\begin{pmatrix} A_i^\top  \Xi_{22} & Z_i \end{pmatrix}\right).$
              \end{tabular}
              }
              \caption{Combination step for parallel Kalman filter.  On the left, there are the covariance-based equations, and on the right there are the proposed square-root versions.}
              \label{tab:sqrt_combination_filter}
            \end{center}
            \end{table}
            \end{lemma} 
            \begin{proof}  In Lemma~\ref{lem:sqrt_combination_filter}, we start by  using matrix inversion lemma (see, e.g.,~\cite{anderson1979optimal}) to recover the square-root version of $ (I_{n_x} + J_j C_i)^{-1} $ appearing in the standard combination formulas as follows:
            \begin{equation} \label{eq:sqrt-com-inversion}
            \begin{split}
              (I_{n_x} + J_j C_i)^{-1} 
              &= (I_{n_x} + J_j U_i I_{n_x} U_i^\top)^{-1} = I_{n_x} - J_j U_i (I_{n_x} + U_i^\top J_j U_i)^{-1} U_i^\top.
            \end{split}
            \end{equation}
            Then by defining the block matrix~\eqref{eq:sqrt-filter-com-2} and applying~\eqref{eq:CCT} we have
            \begin{equation}
            \begin{split} \label{eq:sqrt-filter-com-1}
              \begin{pmatrix}
              U_i^\top Z_j & I_{n_x} \\
              Z_j   &     0
              \end{pmatrix}  & \begin{pmatrix}
              U_i^\top Z_j  & I_{n_x} \\
              Z_j        & 0
              \end{pmatrix}^\top = \begin{pmatrix}
              U_i^\top J_j U_i + I_{n_x}   & U_i^\top J_j \\
              J_j U_i             &    J_j
              \end{pmatrix}.
            \end{split}
            \end{equation}
            Then, because
            \begin{equation}
            \begin{split}
               \begin{pmatrix}
                  \Xi_{11} & 0 \\
                  \Xi_{21} & \Xi_{22}
               \end{pmatrix}&
               \begin{pmatrix}
                  \Xi_{11} & 0 \\
                  \Xi_{21} & \Xi_{22}
               \end{pmatrix}^\top = 
               \begin{pmatrix}
                 \Xi_{11} \Xi_{11}^\top & \Xi_{11} \Xi_{21}^\top \\
                 \Xi_{21} \Xi_{11}^\top & \Xi_{21} \Xi_{21}^\top + \Xi_{22} \Xi_{22}^\top
               \end{pmatrix},
            \end{split}
            \end{equation}
            we can recover the following equations:
            \begin{equation}
            \begin{split}\label{eq:Xis}
              \Xi_{11} \Xi_{11}^\top &= U_i^\top J_j U_i + I_{n_x}, \quad \Xi_{21} \Xi_{11}^\top = J_j U_i, \quad \Xi_{21} \Xi_{21}^\top + \Xi_{22} \Xi_{22}^\top = J_j.
            \end{split}
            \end{equation}
            Hence we can write:
            \begin{equation}
              J_j U_i (I_{n_x} + U_i^\top J_j U_i)^{-1} U_i^\top = \Xi_{21} \Xi_{11}^{-1} U_i^\top,
            \end{equation}
            so that~\eqref{eq:sqrt-com-inversion} can be rewritten as:
            \begin{equation}
              (I_{n_x} + J_j C_i)^{-1}
              = I_{n_x} - \Xi_{21} \Xi_{11}^{-1} U_i^\top.
            \end{equation}
            We also know that
            \begin{equation} \label{eq:sqrt-filter-com-3}
              (I_{n_x} + C_i J_j)^{-1}
              = [ (I_{n_x} + J_j C_i)^{-1} ]^\top
              = I_{n_x} - U_i \Xi_{11}^{-\top} \Xi_{21}^\top,
            \end{equation}
            so that~\eqref{eq:sqrt-filter-com-3} enables us to express $A_{ij}$ as follows:
            \begin{equation}
              A_{ij} = A_j A_i - A_j U_i \Xi_{11}^{-\top} \Xi_{21}^\top A_i.
            \end{equation}
            For $b_{ij}$ and $C_{ij}$ we need the following term:
            \begin{equation} \label{eq:sqrt-filter-com-4}
            \begin{split}
              (I_{n_x} + C_i J_j)^{-1} C_i &= C_i (I_{n_x} + J_j C_i)^{-1} \\
              &= U_i U_i^\top - U_i U_i^\top J_j U_i (I_{n_x} + U_i^\top J_j U_i)^{-1} U_i^\top \\
              &= U_i (I_{n_x} - U_i^\top J_j U_i (I_{n_x} + U_i^\top J_j U_i)^{-1}) U_i^\top \\
              &= U_i ((I_{n_x} + U_i^\top J_j U_i) (I_{n_x} + U_i^\top J_j U_i)^{-1} \\
              &- U_i^\top J_j U_i (I_{n_x} + U_i^\top J_j U_i)^{-1}) U_i^\top \\
              &= U_i (I_{n_x} + U_i^\top J_j U_i)^{-1} U_i^\top = U_i [\Xi_{11} \Xi_{11}^\top]^{-1} U_i^\top,
            \end{split}
            \end{equation}
            which gives the following factorization:
            \begin{equation} \label{eq:sqrt-filter-com-5}
            \begin{split}
              A_j (I_{n_x} + C_i J_j)^{-1} C_i A_j^\top = [A_j U_i \Xi_{11}^{-\top}] [A_j U_i \Xi_{11}^{-\top}]^\top.
            \end{split}
            \end{equation}
            Equation~\eqref{eq:sqrt-filter-com-5} can be summed with $C_j$ by using the triangularization property~\eqref{eq:CCT} to give $C_{ij} = A_j (I_{n_x} + C_i J_j)^{-1} C_i A_j^\top + C_j$.  For $J_{ij}$ we will further need:
            \begin{equation} \label{eq:sqrt-filter-com-6}
            \begin{split}
              (I_{n_x} + J_j C_i)^{-1} J_j
              &= (I_{n_x} - \Xi_{21} \Xi_{11}^{-1} U_i^\top) J_j \\
              &= J_j - \Xi_{21} \Xi_{11}^{-1} [\Xi_{21} \Xi_{11}^\top]^\top \\
              &= \Xi_{21} \Xi_{21}^\top + \Xi_{22} \Xi_{22}^\top - \Xi_{21} \Xi_{21}^\top =  \Xi_{22} \Xi_{22}^\top.
            \end{split}
            \end{equation}
            Using~\eqref{eq:sqrt-filter-com-6} will give us $A_i^\top (I_{n_x} + J_j C_i)^{-1} J_j A_i= [A_i^\top \Xi_{22}] [A_i^\top \Xi_{22}]^\top$ which can be summed with $J_i$ using the triangulation property~\eqref{eq:CCT} again to give $J_{ij}$.
            \end{proof}
            We summarize the square-root version of the parallel filtering step in Algorithm~\ref{Alg:parallel-filtering-sqrt-std}.
            \begin{algorithm}[tbhp]
            \caption{\textsc{Standard} or \textsc{Square-root} versions of parallel Kalman filter}\label{Alg:parallel-filtering-sqrt-std}
            \begin{algorithmic}[1]
            \renewcommand{\algorithmicrequire}{\textbf{Input:}}
            \renewcommand{\algorithmicensure}{\textbf{Output:}}
            \REQUIRE SSM with parameters in~\eqref{eq:linear-model}, measurements $y_{1:n}$, standard version with initial mean and covariance $\{m_0, P_0\}$ or square-root version with initial mean and Cholesky of covariance $\{m_0, N_0\}$.
            \ENSURE Filtering results $\{m^f_{0:n}, P^f_{0:n}\}$ or $\{m^f_{0:n}, N^f_{0:n}\}$.
            \STATE
            \COMMENT {Initialization}
            \FOR[Compute in parallel]{$k\gets1$ \textbf{to} $n$}
            		\STATE {Set $a_k = \{A_k, b_k, C_k \, , \eta_k, J_k \}$ or $a_k = \{A_k, b_k, U_k, \eta_k, Z_k\}$ according to Lemma~\ref{lem:sqrt_filter_k1} in \textsc{Standard} or \textsc{Square-root} versions.}
            \ENDFOR
            \STATE{\COMMENT {Combination}}
            \STATE {Set $(\{A_k, b_k, C_k \, \text{or} \, U_k, \eta_k, J_k \, \text{or} \, Z_k\})_{k=1}^n = \mathrm{AssociativeScan}(\otimes, (a_k)_{k=1}^n)$,\\
            where $\otimes$ defined in~Lemma~\ref{lem:sqrt_combination_filter} in \textsc{Standard} or \textsc{Square-root} versions.}
            \STATE{\COMMENT{Filtering results}}
             \STATE{Set $\{m^f_0, P^f_0 \, \text{or} \, N^f_0\} = \{m_0, P_0 \, \text{or} \, N_0\}$ and extract $\{m^f_{1:n}, P^f_{1:n} \, \text{or} \, N^f_{1:n}\} = \{b_{1:n}, C_{1:n} \, \text{or} \, U_{1:n}\}$ in \textsc{Standard} or \textsc{Square-root} versions.}
            \end{algorithmic}
            \end{algorithm}
\subsection{Parallel square-root smoother for linear systems}
\label{subsec:parallel-sqrt-smoother}
The aim in this section is to review the standard parallel RTS smoother formulation and develop its square-root version. Similarly to the parallel filtering case, we need to develop square-root formulations of both the initialization and combination steps. 

To this end, we consider that we have the SSM with parameters in~\eqref{eq:linear-model}. Also, we assume that we have obtained the filtering results, $\{m^f_{k}$, $P^f_{k}\}$, in the standard version or $\{m^f_{k}$, $N^f_{k}\}$ in the square-root version such that $P^f_k = N^f_k [N^f_k]^\top$  as in Algorithm~\ref{Alg:parallel-filtering-sqrt-std}. 

In the standard version, the initialization of parameters can be obtained by parametrizing the element $a_k =  p(x_k \mid y_{1:k}, x_{k+1}) $ as a set $\{E_k,g_k,L_k\}$ such that $p(x_k \mid y_{1:k}, x_{k+1}) = N(x_k; E_k x_{k+1} + g_k, L_k)$~\cite[Lemma 9]{sarkka2020temporal}. We now derive the square-root version of the smoothing initialization elements, which we define as
\begin{equation}
    \Pi^s_k = \{E_k, g_k, D_k \},
\end{equation}
where $L_k = D_k D_k^\top$. The square-root version of the initialization parameters for parallel Kalman smoother is given in the following lemma.
\begin{lemma}[Initialization step of square-root parallel RTS smoother] \label{lem:init_smoother}
Table~\ref{tab:init_smoother} presents square-root formulations of~\cite{sarkka2020temporal} to find the initialization parameters of parallel Kalman smoother for $0 \leq k < n $, where we define the following matrix using $\mathrm{Tria}$:
 \begin{equation}
\begin{split} \label{eq:smooth-init-2}
   \begin{pmatrix}
      \Phi_{11} & 0 \\
      \Phi_{21} & \Phi_{22}
   \end{pmatrix}
   = \mathrm{Tria}\left(
      \begin{pmatrix}
      F_k N^f_k & S_{\Lambda_k} \\
      N^f_k & 0
   \end{pmatrix}
     \right).
\end{split}
\end{equation}

\begin{table}[tbhp]
\begin{center}
  \begin{tabular}{c|c} 
       \textbf{Standard version}~\cite[Lemma 9]{sarkka2020temporal}& \bf Square-root version  \\ [0.1cm]
       $E_k = P_k F_k^\top  (F_k P^f_k F_k^\top + \Lambda_{k-1} )^{-1}$, & $ E_k = \Phi_{21} \Phi_{11}^{-1}$,\\ [0.2cm]
        $g_k = m^f_k - E_k (F_k m^f_k + c_k)$, &  $g_k = m^f_k - E_k(F_k m^f_k + c_k)$,\\ [0.2cm]
        $L_k = P^f_k - E_k F_k P^f_k$, & $D_k = \Phi_{22}$.
  \end{tabular}
  \caption{Initialization step in parallel RTS smoother for $0 \leq k < n $.  On the left there are the covariance-based formulations, and on the right there are the proposed square-root version. For the last step, the initialization parameters are $E_n = 0$, $g_n = m^f_n$, and $ D_n = N^f_n$.}
  \label{tab:init_smoother}
\end{center}
\end{table}
\end{lemma}

\begin{proof}
In Lemma~\ref{lem:init_smoother}, we apply~\eqref{eq:CCT} in~\eqref{eq:smooth-init-2} then we have
\begin{equation}
\begin{split} \label{eq:smooth-init-1}
   \begin{pmatrix}
      F_k N^f_k & S_{\Lambda_k} \\
      N^f_k & 0
   \end{pmatrix}&
   \begin{pmatrix}
      F_k N^f_k & S_{\Lambda_k} \\
      N^f_k & 0
   \end{pmatrix}^\top = 
   \begin{pmatrix}
      F_k P^f_k F_k^\top + \Lambda_k & F_k P^f_k \\
      P_k F_k^\top & P^f_k
   \end{pmatrix},
\end{split}
\end{equation}
By considering $L_k = D_k D_k^\top$ and applying~\eqref{eq:CCT} in~\eqref{eq:smooth-init-1} and~\eqref{eq:smooth-init-2} the square-root parameters can be computed. 
\end{proof}
Now we can describe the binary associative operator $\otimes$ for the smoothing step. In the standard version, the result of the associative smoothing operator is defined by $\{E_i,g_i,L_i\} \otimes \{E_j,g_j,L_j\} \coloneqq \{E_{ij},g_{ij},L_{ij}\}$~\cite[lemma 10]{sarkka2020temporal}. We now aim to find the square-root version of the smoothing combination terms, that is, $\{E_{ij}, g_{ij}, D_{ij}\}$, where
 \begin{equation}
            \begin{split}
                 \{E_i, g_i, D_i\} \otimes \{E_j, g_j, D_j\} &=  \{E_{ij}, g_{ij}, D_{ij}\}.
            \end{split}
            \end{equation}
\begin{lemma}[Combination step of square-root parallel RTS smoother] \label{lem:combination_smoother}
Table~\ref{tab:combination_smoother} presents square-root formulations of~\cite{sarkka2020temporal} in the combination step of parallel Kalman smoother.
\end{lemma}

\begin{table}[tbhp]
\begin{center}
    \begin{tabular}{c|c} 
       \textbf{Standard version}~\cite[Lemma 10]{sarkka2020temporal}& \bf Square-root version  \\ [0.1cm]
       $E_{ij} = E_i E_j$, & $ E_{ij} = E_i E_j$,\\ [0.2cm]
        $g_{ij} = E_i g_j + g_i$, &  $g_{ij} = E_i g_j + g_i$,\\ [0.2cm]
        $L_{ij} = E_i L_j E_i^\top + L_i$, & $D_{ij}
    = \mathrm{Tria}\left( \begin{pmatrix} E_{i} D_j &  D_i \end{pmatrix} \right)$.
   \end{tabular}
   \caption{Combination step for parallel Kalman smoother.  On the left are covariance-based formulations, and on the right are the proposed square-root version.}
   \label{tab:combination_smoother}
\end{center}
\end{table}

\begin{proof}
   In Lemma~\ref{lem:combination_smoother}, it suffices to use~\eqref{eq:CCT} and find the square-root of $L_{ij} =  E_i L_j E_i^{\top} + L_i$, where $L_i = D_i D_i^\top$ and $L_j = D_j D_j^\top$.
\end{proof}
The smoothing results can then be computed by applying the parallel associative scan algorithm reversed in time. Algorithm~\ref{alg: parallel-smoothing-sqrt-std} presents the procedure of computing the initialization and combination of the parallel standard and square-root smoother.
\begin{algorithm}[t]
\caption{\textsc{Standard} or \textsc{Square-root} versions of parallel RTS smoother}\label{alg: parallel-smoothing-sqrt-std}
\begin{algorithmic}[1]
\renewcommand{\algorithmicrequire}{\textbf{Input:}}
\renewcommand{\algorithmicensure}{\textbf{Output:}}
\REQUIRE SSM with parameters in~\eqref{eq:linear-model}, $\Lambda_k = S_{\Lambda_k} S_{\Lambda_k}^\top$, and $\Omega_k = S_{\Omega_k} S_{\Omega_k}^\top$. Filtering results in  standard version, $\{m^f_{0:n}$, $P^f_{0:n}\}$, or square-root version, $\{m^f_{0:n}$, $N^f_{0:n}\}$.
\ENSURE Smoothing results $\{m^s_{0:n}, P^s_{0:n}\}$ or $\{m^s_{0:n}, N^s_{0:n}\}$
\STATE
\COMMENT {Initialization}
  \FOR[Compute in parallel]{$k\gets 0$ \textbf{to} $n$}
    \STATE {Set $a_k = (\{E_k, g_k, L_k \, \text{or} \, D_k\})_{k=1}^n $  using Lemma~\ref{lem:init_smoother} in \textsc{Standard} or \textsc{Square-root} versions.}
   \ENDFOR
\STATE{\COMMENT {Combination}}
\STATE {Set $(\{E_k, g_k, L_k \, \text{or} \, D_k\})_{k=1}^n = \mathrm{ReverseAssociativeScan}(\otimes,(a_k)_{k=0}^n)$, \\where $\otimes$ is defined in Lemma~\ref{lem:combination_smoother} in \textsc{Standard} or \textsc{Square-root} versions.}
\STATE{\COMMENT{Smoothing results}}
 \STATE{Extract $(m^s_{0:n}, P^s_{0:n} \, \text{or} \, N^s_{0:n}) = (g_{0:n}, L_{0:n} \, \text{or} \,  D_{0:n})$ in \textsc{Standard} or \textsc{Square-root} versions.}
\end{algorithmic}
\end{algorithm}

\section{Parallel standard and square-root iterated filter and smoother based on GSLR posterior linearization}
\label{sec:parallel-filter-smoother-gslr}
In this section, we will generalize the parallel iterated filtering and smoothing methodology of~\cite{yaghoobi2021parallel} to generalized statistical linear regression (GSLR). Then we will derive a square-root version of the parallel GSLR filtering and smoothing methods.

As shown in~\cite{garcia2016iterated,tronarp2018iterative}, a single iteration of the GSLR-based posterior linearization filtering and smoothing consists of:
\begin{enumerate}
\item Using the smoothing result from the previous iteration to linearize the dynamic and measurement models of the system at all time steps via GSLR.
\item Running a Kalman filter and smoother on the linearized system.
\end{enumerate}
In order to parallelize the above steps, we need to parallelize the linearization as well as the Kalman filter and smoother. Below, we provide details on how to find the linearized parameters for a general SSM~\eqref{eq:ss-model} in standard or square-root versions. We can then apply Algorithms~\ref{Alg:parallel-filtering-sqrt-std} and~\ref{alg: parallel-smoothing-sqrt-std} to obtain the parallel iterated method. 

\subsection{Standard and square-root forms of generalized statistical linear regression}
We now describe the procedure of finding a linearized representation of the SSM~\eqref{eq:ss-model} for standard and square-root versions. In order to perform the procedure iteratively, we consider a sequence of such representations in standard form as
\begin{subequations} \label{eq:linearized-model-at-time-i}
\begin{align}
        x_k &\approx F^{(i)}_{k-1} x_{k-1} + c^{(i)}_{k-1} + q^{(i)}_{k-1}, &&  \, q^{(i)}_{k} \sim \mathcal{N}(0, \Lambda^{(i)}_k), \\
        y_k &\approx H^{(i)}_k x_k + d^{(i)}_k + v^{(i)}_k, &&  \, v^{(i)}_{k} \sim \mathcal{N}(0, \Omega^{(i)}_k).
\end{align}
\end{subequations}
We define the parameters appearing at iteration $i$ as: 
\begin{equation}  \label{eq:linear-params}
   \Gamma_{1:n}^{(i)} = \{F_{0:n-1}^{(i)}, c_{0:n-1}^{(i)}, \Lambda_{0:n-1}^{(i)}, H_{1:n}^{(i)}, d_{1:n}^{(i)}, \Omega_{1:n}^{(i)}\}.
\end{equation}
Given the current set of linearized parameters $\Gamma_{1:n}^{(i)}$, we can then compute the parameters at the next iteration, $i+1$, by applying the GSLR method to the transition and observation densities $p(x_{k} \mid x_{x-1})$ and $p(y_k \mid x_k)$, respectively, over the current linearized estimate of the smoothing distribution marginals $p(x_k \mid y_{1:n}, \Gamma_{1:n}^{(i)})$. For all $k \geq 1$, we compute the following quantities in parallel:
\begin{subequations}\label{eq:linear-params-trans}
    \begin{align}
        \begin{split} \label{eq:linear-params-trans-1}
            F_{k-1}^{(i+1)} &= \mathbb{C}[\mathbb{E}[x_k \mid x_{k-1}],x_{k-1}]^\top \mathbb{V}[x_{k-1}]^{-1},
        \end{split}\\
        \begin{split}\label{eq:linear-params-trans-2}
            c_{k-1}^{(i+1)} &= \mathbb{E}[\mathbb{E}[x_k \mid x_{k-1}]] - F^{(i+1)}_{k-1} \mathbb{E}[x_{k-1}], 
        \end{split}\\
        \begin{split}\label{eq:linear-params-trans-3}
            \Lambda_{k-1}^{(i+1)} &= \mathbb{E}[\mathbb{V}[x_k\mid x_{k-1}]]+\mathbb{V}[\mathbb{E}[x_k\mid x_{k-1}]] - F^{(i+1)}_{k-1} \mathbb{V}[x_{k-1}] [F^{(i+1)}_{k-1}]^\top,
        \end{split}
    \end{align}
\end{subequations}
where the expectations are taken over $p(x_{k-1} \mid y_{1:n}, \Gamma_{1:n}^{(i)})$, and 
\begin{subequations} \label{eq:linear-params-obs}
    \begin{align}
        \begin{split}
             H_{k}^{(i+1)} &= \mathbb{C}[\mathbb{E}[y_k \mid x_{k}],x_{k}]^\top \mathbb{V}[x_k]^{-1},
        \end{split}\\
        \begin{split}
             d_{k}^{(i+1)} &= \mathbb{E}[\mathbb{E}[y_k \mid x_{k}]] - H^{(i+1)}_{k} \mathbb{E}[x_{k}],
        \end{split}\\
        \begin{split}
            \Omega_{k}^{(i+1)} &= \mathbb{E}[\mathbb{V}[y_k\mid x_k]]+\mathbb{V}[\mathbb{E}[y_k\mid x_k]]  - H^{(i+1)}_{k} \mathbb{V}[x_k] [H^{(i+1)}_{k}]^\top,
        \end{split}
    \end{align}
\end{subequations}
where the expectations are now taken over $p(x_{k} \mid y_{1:n}, \Gamma_{1:n}^{(i)})$. Importantly, the approximate smoothing marginals are Gaussian and are therefore fully represented by their corresponding means $m^{s, (i)}_{0:n}$ and covariances $P^{s, (i)}_{0:n}$. While the conditional expectations and covariances involved in~\eqref{eq:linear-params-trans} and~\eqref{eq:linear-params-obs} are not necessarily tractable, we can then approximate them by using a Taylor expansion or sigma-point-based integration methods. Crucially, this results in a method that is fully independent across all steps $k$, so that it can be fully parallelized across all steps.

We now consider the problem of deriving a square-root version of the linearized parameters $\Gamma$ in~\eqref{eq:linear-params}. We define the square-root linearization parameters, $\Gamma^{\textrm{sqrt}}$, in each iteration, $i$, as: 
\begin{equation} \label{eq:Gamma-s}
    \Gamma^{\textrm{sqrt}, (i)}_{1:k} = \{F_{0:k-1}^{(i)}, c_{0:k-1}^{(i)}, S_{\Lambda_{0:k-1}}^{(i)}, H_{1:k}^{(i)}, d_{1:k}^{(i)}, S_{\Omega_{1:k}}^{(i)}\},
\end{equation}
where $S_{\Lambda_k}$ and $S_{\Omega_k}$ are the square-root matrices of $\Lambda_k$ and $\Omega_k$, respectively. Below, we derive the parameters of $\Gamma^{\textrm{sqrt}}$, using Taylor series expansion and sigma-point methods. In both methods, we assume that the Cholesky decompositions of the conditional covariances of transition and observation densities, $\mathbb{V}[x_k \mid x_{k-1}]$ and  $\mathbb{V}[y_k \mid x_k]$, defined as $\mathbb{S}_x[x_k \mid x_{k-1}]$ and $\mathbb{S}_y[y_k \mid x_k]$, respectively, are tractable.

\subsection{Square-root iterated Taylor series expansion}\label{subsec:sqrt-taylor}
Iterated Taylor series expansions use first-order Taylor approximations to linearize nonlinear functions. This linearization occurs around the previous smoothing mean $m^s_k$, which is the best available result~\cite{garcia2016iterated}. Using this method, the conditional mean and square root of the covariance in~\eqref{eq:linear-params-trans} can be approximated as follows
\begin{equation} \label{eq:taylor-approx-cm-cv}
    \begin{split}
       \mathbb{E}[x_k \mid x_{k-1}] &\approx \mathbb{E}[x_k \mid x_{k-1}](m^s_{k-1}) + {\nabla_{x}}\mathbb{E}[x_k \mid x_{k-1}](m^s_{k-1})(x_{k-1} - m^s_{k-1}),\\
       \mathbb{S}_x[x_k \mid x_{k-1}] &\approx \mathbb{S}_x[x_k \mid x_{k-1}](m^s_{k-1}).
    \end{split}
\end{equation}
Then, by applying~\eqref{eq:taylor-approx-cm-cv} in~\eqref{eq:linear-params-trans}, the square-root version of the first-order Taylor series expansion method will be as follows:
\begin{equation*}
\begin{split}
    F_{k-1} &= \nabla_x\mathbb{E}[x_k \mid x_{k-1}]( m^s_{k-1}),\quad c_{k-1} = \mathbb{E}[x_k \mid x_{k-1}](m^s_{k-1}) - F_{k-1} \mathbb{E}[x_{k-1}], \\
    S_{{\Lambda}_{k-1}} &= \mathbb{S}_x[x_k \mid x_{k-1}](m^s_{k-1}).
\end{split}
\label{eq:square_root_taylor_FcS}%
\end{equation*}
The same procedure can be applied to find the parameters $(H_k, d_k, S_{\Omega})$. The conditional mean and square-root covariance needed for~\eqref{eq:linear-params-obs} are:
\begin{equation} \label{eq:taylor-approx-cm-cv-obs}
    \begin{split}
       \mathbb{E}[y_k \mid x_k] &\approx \mathbb{E}[y_k \mid x_k](m^s_k) + {\nabla_{x}}\mathbb{E}[y_k \mid x_k](m^s_k)(x_k - m^s_k),\\
       \mathbb{S}_x[y_k \mid x_k] &\approx \mathbb{S}_x[y_k \mid x_k](m^s_k).
    \end{split}
\end{equation}
Now by substituting~\eqref{eq:taylor-approx-cm-cv-obs} in~\eqref{eq:linear-params-obs}, the square-root linearized parameters are: 
\begin{equation*}
\begin{split}
    H_k &= \nabla_x\mathbb{E}[y_k \mid x_k](m^s_k),\quad d_k = \mathbb{E}[y_k \mid x_k](m^s_k) - H_k \mathbb{E}[x_k],\quad S_{{\Omega}_{k}} = \mathbb{S}_y[y_k \mid x_k](m^s_k).
\end{split}
\label{eq:square_root_taylor}
\end{equation*}

\subsection{Square-root iterated sigma-point linearization}\label{subsec:sqrt-sigma}
In general, sigma-point linearization methods aim to approximate the expectation of a nonlinear function using a set of sigma points, $\chi_n$, and their associated weights, $w_n$ as~\cite{tronarp2018iterative}:
\begin{equation} \label{eq:approx_sigmapoint}
    \begin{split}
       \mathbb{E}[y]& \approx \sum_n w_n \mathbb{E}[y \mid x=\chi_{n}].
    \end{split}
\end{equation}
When the random variable $x$ is Gaussian, a wide range of sigma-point methods is available~\cite{sarkka2013bayesian} for this purpose.

 Here we consider $s$ sigma points $\mathcal{X}_1,\ldots,\mathcal{X}_s$, their associated positive mean weights $w^m _1,\ldots,w^m _s$, and positive covariance weights $w^c _1,\ldots,w^c _s$. The idea of the square-root iterated sigma-points linearization method is to obtain the required expectations according to posterior means and square-root covariances of the smoothing density $p(x_k \mid y_{1:n})$~\cite{ arasaratnam2009cubature,arasaratnam2011cubature}. It is worth noting that there are different choices for sigma points and weights which lead to different sigma-point methods~\cite{kokkala2015sigma}. 

First, we want to compute the parameters $(F_{k-1},c_{k-1},S_{\Lambda_{k-1}})$. To this end, we need both the mean and the square root of the covariance of the smoothing density, $m^s_k$ and $N^s_k$. Then, we obtain the transformed sigma-point as $\mathcal{Z}_{x,i} = \mathbb{E}[x_k \mid x_{k-1}](\mathcal{X}_i)$ and consequently define:
\begin{equation} \label{eq:z-bar}
    \bar{z}_x = \mathbb{E}[\mathbb{E}[x_k \mid x_{k-1}]] \approx \sum_{i=1}^{s} w^m_i \mathcal{Z}_{x,i}.
\end{equation}
Then, by applying~\eqref{eq:approx_sigmapoint} in~\eqref{eq:linear-params-trans} and using~\eqref{eq:z-bar} we have:
\begin{equation*} \label{eq:F-c-sigma-point-params}
\begin{split}
F_{k-1} &\approx \sum_{i=1}^{s} w^c_i (\mathcal{Z}_{x,i}- \bar{z}_x)(\mathcal{X}_i - m^s_{k-1})^\top (N^s_{k-1} [N^s_{k-1}]^\top)^{-1},\quad c_{k-1} \approx \bar{z}_x - F_{k-1} m^s_{k-1}.
\end{split}
\end{equation*}
In order to find the square root of the covariance as per~\eqref{eq:linear-params-trans-3} we define: 
\begin{equation}
\begin{split}
    Z_x &= 
    \begin{bmatrix}
        \sqrt{w_1^c} (\mathcal{Z}_{x,1} - \bar{z}_x) & \cdots & \sqrt{w_s^c} (\mathcal{Z}_{x,s} - \bar{z}_x)
    \end{bmatrix}.
\end{split}
\label{eq:transformed_z}
\end{equation}
Now, by defining $\mathcal{S}_{x,i} = \mathbb{S}_x[x_k \mid x_{k-1}](\mathcal{X}_i)$, for $i=1, \ldots,s$, and applying~\eqref{eq:approx_sigmapoint} in~\eqref{eq:linear-params-trans-3}, the latter equation can be approximated as:
\begin{equation} \label{eq:approx-Lambda}
\begin{split}
    \Lambda_{k-1} &\approx   \sum_{i=1}^{s} \left( w^c_i \mathcal{S}_{x,i}  \mathcal{S}_{x,i}^\top\right) + Z_x Z_x^\top - F_{k-1} N^s_{k-1} [N^s_{k-1}]^\top F_{k-1}^\top.
\end{split}
\end{equation}
To find the square root of $\Lambda_k$, first, we find the square root of the first two parts on the right-hand side of~\eqref{eq:approx-Lambda} as follows:
\begin{equation} \label{eq:lambda-prime}
S_{\Lambda_{k-1}}^\prime=\mathrm{Tria}(\sqrt{w^c_1} \mathcal{S}_{x,1}, \ldots, \sqrt{w^c_s} \mathcal{S}_{x,s}, Z_x). 
\end{equation}
Then, using $\mathrm{DownDate}$ function (see Section~\ref{subec:background-tria}), the square-root of $\Lambda_{k-1}$ will be
\begin{equation}\label{eq:sqr-lambda-sigma-point}
 \begin{split}
S_{\Lambda_{k-1}} \approx \mathrm{DownDate}(S_{\Lambda_{k-1}}^\prime, F_{k-1} N^s_{k-1}).
\end{split}
\end{equation}
The same procedure can be applied in order to obtain a square-root version of the linearization parameters $(H_k, d_k, S_{\Omega_k})$ of~\eqref{eq:linear-params-obs}. To do so, we consider $ \mathcal{Z}_{y,i} = \mathbb{E}[y_k \mid x_k](\mathcal{X}_{i})$, and $\mathcal{S}_{y,i} = \mathbb{S}[y_k \mid x_k](\mathcal{X}_{i})$ for $i=1, \ldots,s$. Then, by using the mean, $\bar{z}_y$, and defining $Z_y$ as
\begin{equation*}
     \bar{z}_y = \mathbb{E}[\mathbb{E}[y_k \mid x_k]] \approx \sum_{i=1}^{s} w^m_i \mathcal{Z}_{y,i}, \quad Z_y = \begin{bmatrix}
     \sqrt{w_1^c}(\mathcal{Z}_{y,1} - \bar{z}_y) & \cdots& \sqrt{w_s^c}(\mathcal{Z}_{y,s} - \bar{z}_y)
     \end{bmatrix},
\end{equation*}
 we can compute the linearization parameters $(H_k, d_k, S_{\Omega_k})$ as:
\begin{equation}\label{eq:sqr-obs-linearized-param-sigma-point}
 \begin{split}
 H_{k} &\approx  \sum_{i=1}^{s} w^c_i (\mathcal{Z}_{y,i}- \bar{z}_y)(\mathcal{X}_i - m^s_k)^\top (N^s_{k} [N^s_k]^\top)^{-1}, \quad d_{k}  \approx \bar{z}_y - H_{k} m^s_{k},\\
S_{\Omega_k} &\approx \mathrm{DownDate}(\mathrm{Tria}(\sqrt{w^c_1} \mathcal{S}_{y,1}, \cdots, \sqrt{w^c_s} \mathcal{S}_{y,s}, Z_y), H_k N^s_k).
\end{split}
\end{equation}

\subsection{Parallel standard and square-root algorithms for general state-space models}
We can now find standard and square-root  forms of the parallel iterated filter and smoother based on GSLR linearization. To be specific, a single iteration of these methods consists of the following two steps:
\begin{enumerate}
    \item Using standard or square-root results of the smoother from the previous iteration in GSLR to linearize the dynamic and measurement models of the nonlinear system on all time steps $k = 1,\ldots,n$. The linearization can be performed based on Taylor series expansion or sigma-point methods. 
    \item Running the standard or square-root versions of Kalman filter and smoother on the linearized system in Algorithms~\ref{Alg:parallel-filtering-sqrt-std} and~\ref{alg: parallel-smoothing-sqrt-std}.
\end{enumerate}
Defining the standard and square-root methods allows us to iterate the filtering and smoothing algorithms until convergence. This recovers the optimal nominal trajectory, defined here as $\mathfrak{n}^*_{0:n} \coloneqq \{m^s_{0:n}, P^s_{0:n}\}$ for the standard method and $\mathfrak{n}^*_{0:n} = \{m^s_{0:n}, N^s_{0:n}\}$ for the square-root method. This is summarized in Algorithm~\ref{alg:fixed_point_iteration}.

\begin{algorithm}[tbhp]
\caption{Parallel standard or square-root iterated filter and smoother based on GSLR linearization}\label{alg:fixed_point_iteration}
\begin{algorithmic}[1]
\renewcommand{\algorithmicrequire}{\textbf{Input:}}
\renewcommand{\algorithmicensure}{\textbf{Output:}}
\REQUIRE Initial smoothing results for the whole trajectory at iteration $i=0$ in standard version, $\{m^{s, (0)}_{0:n}, P^{s, (0)}_{0:n}\}$, or square-root version, $\{m^{s, (0)}_{0:n}, N^{s, (0)}_{0:n}\}$, and the number of iterations $M$.
\ENSURE Smoothing results at the final iteration in the standard version,  $\{m^{s, (M)}_{0:n}, P^{s, (M)}_{0:n}\}$, or square-root version $\{m^{s, (M)}_{0:n}, N^{s, (M)}_{0:n}\}$.
\FOR{$i\gets 1$ \textbf{to} $M$}
\STATE
    \COMMENT{Compute linearization parameters}
    \FOR[Compute in parallel]{$k\gets 1$ \textbf{to} $n$}
    \IF{\textsc{Standard version}}
    \STATE {Compute $F_{k-1}^{(i)}$, $c_{k-1}^{(i)}$, and $\Lambda_{k-1}^{(i)}$ using~\eqref{eq:linear-params-trans} and $\{m^{s,(i-1)}_{k-1}, P^{s,(i-1)}_{k-1}\}$.}
		\STATE {Compute $H_k^{(i)}$, $d_k^{(i)}$, and $\Omega_k^{(i)}$ using~\eqref{eq:linear-params-obs} and $\{m^{s,(i-1)}_{k},P^{s,(i-1)}_k\}$.}
    \ELSIF{\textsc{square-root version}}
        \STATE {Compute $\Gamma^{\textrm{sqrt}, (i)}_{1:k} = \{F_{0:k-1}^{(i)}, c_{0:k-1}^{(i)}, S_{\Lambda_{0:k-1}}^{(i)}, H_{1:k}^{(i)}, d_{1:k}^{(i)}, S_{\Omega_{1:k}}^{(i)}\}$ by using a Taylor series expansion (Sec.~\ref{subsec:sqrt-taylor}) or sigma-point methods (Sec.~\ref{subsec:sqrt-sigma}).} %
    \ENDIF
    \ENDFOR
    \STATE {Run standard or square-root version of Algorithm~\ref{Alg:parallel-filtering-sqrt-std} on the linearized model to compute $\{m^{f,(i)}_{0:n}, N^{f,(i)}_{0:n}\}$.}
    \STATE {Run standard or square-root version of Algorithm~\ref{alg: parallel-smoothing-sqrt-std} on the linearized model to compute $\{m^{s,(i)}_{0:n}, N^{s,(i)}_{0:n}\}$.}
\ENDFOR
\end{algorithmic}
\end{algorithm}

\section{Parameter estimation}
\label{sec:sys-identification}

    In order to derive parallel state estimation methods in the previous sections, we have assumed that the SSM from which the observations $y_{1:n}$ were generated was known exactly. We now consider the case when the model at hand depends on a parameter $\theta$ that we want to estimate:
    \begin{align}
        p(x_k \mid x_{k-1}, \theta), \quad
        p(y_k \mid x_k, \theta), \quad
        p(x_0 \mid \theta).
    \end{align}
    A typical way to do so is via the marginal log-likelihood function, and its derivative with respect to the parameter $\theta$: the score function. In this section, we discuss how both of these can be computed in parallel. A particular attention is given to the score function as computing it naively would incur a memory cost scaling linearly with the number of iterations in Algorithm~\ref{alg:fixed_point_iteration}.
    For the sake of conciseness, in the remainder of this section, we will only consider the covariance formulation of our proposed method, that is, the standard version of Algorithm~\ref{alg:fixed_point_iteration}. However, the analysis remains unchanged if one considers instead the square-root formulation developed in Sections~\ref{subsec:parallel-sqrt-filter},~\ref{subsec:parallel-sqrt-smoother},~\ref{subsec:sqrt-taylor}, and~\ref{subsec:sqrt-sigma}, provided that the covariance operations are all replaced by their square-root equivalents. 
    Similarly, for notational simplicity, we will only consider the case when $n_x = n_y = \mathrm{dim}(\theta) =1$, the general case follows from taking element-wise derivatives instead. 

\subsection{Parallel computation of the (pseudo) log-likelihood}\label{subsec:log-lik}
    For general SSMs, the log-likelihood $\log p(y_{1:n}) = p(y_{1:n} \mid \theta)$ can be computed iteratively using 
    \begin{equation}
        \begin{split}
        \log p(y_{1:n}) 
            &= \sum_{k=1}^n \log \int_{x_k} p(y_k \mid x_k) p(x_k \mid y_{1:k-1}) \mathrm{d}x_k,
        \end{split}\label{eq:log-likelihood}
    \end{equation}
    where we drop the dependency on $\theta$ for simplicity. 
    In the case of linear Gaussian SSMs (LGSSMs), the integrals in~\eqref{eq:log-likelihood} are available in closed form, and we can compute them using the observation predictive means and covariances. 
    As was suggested in~\cite{sarkka2020temporal}, for LGSSMs, this formulation allows to compute the log-likelihood $\log p(y_{1:n})$ in parallel by first computing the filtering distribution at each time step $p(x_k \mid y_{1:n}) = N(x_k; m^f_k, P^f_k)$, then all the terms in~\eqref{eq:log-likelihood} $\ell_k = \log N(y_k; m_k^y, P_k^y)$, where $m_k^y$ and $P_k^y$ are the predictive mean and covariance of $y_k$, and summing them for all $k$'s. 
    The span complexity of computing the filtering distribution at each time step is $O(\log(n))$~\cite{sarkka2020temporal}, the span complexity of computing all the terms $\ell_k$ given the filtering distribution is $O(1)$, and the span complexity of summing all these terms is $O(\log(n))$ too so that the total span complexity of computing the log-likelihood in LGSSMs is $O(\log(n))$. 
    
    In the context of general SSMs, however, we cannot usually obtain analytical expressions for the integrals appearing in~\eqref{eq:log-likelihood}. 
    Instead, these can be approximated by leveraging the same linearization rules that were used to compute the filtering estimates. 
    Specifically, all the terms in~\eqref{eq:log-likelihood} can be approximated from the linearized system~\eqref{eq:linearized-model-at-time-i}, producing a pseudo-log-likelihood $\hat{\ell}(\theta)$ which is an estimate of the log-likelihood $\log p(y_{1:n} \mid \theta)$ of the system. 
    As discussed in Section~\ref{sec:parallel-filter-smoother-gslr}, the linearized parameters of the system~\eqref{eq:linearized-model-at-time-i} are computed based on either iterated Taylor or sigma-points approximations in span complexity $O(\log(n))$, resulting in a total span complexity of $O(\log(n))$ for computing the pseudo-log-likelihood, for both our proposed square-root and standard formulations. %
\subsection{Constant-memory parallel computation of the score function}\label{subsec:gradient-estimation}
    Because the pseudo-log-likelihood $\hat{\ell}(\theta)$ of the nonlinear SSM~\eqref{eq:ss-model} is the log-likelihood $\log p(y_{1:n} \mid \Gamma_{1:n})$ of the linearized model~\eqref{eq:linearized-model-at-time-i}, one could apply the reverse accumulation chain rule, for example, by using automatic differentiation libraries, in order to compute the gradient of $\hat{\ell}$. This was, for example, done in~\cite{corenflos2021gaussian} in the linear case to learn parameters of temporal Gaussian processes. 
    
    However, while this method would result in an efficient algorithm to compute the score function of LGSSMs, it would suffer from several drawbacks if applied to Gaussian-approximated nonlinear SSMs directly. 
    Indeed, in this case, the linearization parameters $\Gamma_k$ depend on the final nominal trajectory $\mathfrak{n}^*_{0:n}$  used during the linearization step. 
    In turn, the final nominal trajectory depends on the previously estimated linearized parameters as well as $\theta$, forming a chain of dependency down to the initial nominal trajectory.
    Furthermore, in order to proceed with reverse accumulation, all the intermediary results involved in computing the pseudo-log-likelihood need to be stored. 
    Because the total number of iterations $M$ required for the nominal trajectory to converge is not necessarily known in advance, the memory cost of naively ``unrolling'' the convergence loop would be unknown at compilation time. 
    However, thanks to the fixed-point structure of Algorithm~\ref{alg:fixed_point_iteration}, we can derive a constant-memory algorithm to compute the gradient of quantities depending on the filtering and smoothing distributions (which include the pseudo-log-likelihood in particular).
    
    When using the mean-covariance approach of Section~\ref{sec:parallel-filter-smoother-gslr}, computing the gradient of $\hat{\ell}$ consists in applying the chain rule to $\log p$:
        \begin{equation}
        \begin{split}
            \frac{\mathrm{d}\hat{\ell}}{\mathrm{d}{\theta}} 
            &= \sum_{k=1}^n \bigg[\frac{\mathrm{d}\hat{\ell}}{\mathrm{d}{F_{k-1}}}\bigg(\frac{\partial F_{k-1}}{\partial \theta} + \frac{\partial F_{k-1}}{\partial m^s_{k-1}} \frac{\partial m^s_{k-1}}{\partial \theta} + \frac{\partial F_{k-1}}{\partial P^s_{k-1}} \frac{\partial P^s_{k-1}}{\partial \theta} \bigg) + \ldots \bigg].
        \end{split}
        \label{eq:gradient}
    \end{equation}
    To compute the full derivative, we need to compute the following quantities
    \begin{enumerate}
        \item $\frac{\partial F_{k-1}}{\partial \theta}$, $\frac{\partial F_{k-1}}{\partial m^s_{k-1}}$, $\frac{\partial F_{k-1}}{\partial P^s_{k-1}}$: this is easily done by differentiating through the Taylor or sigma-points linearization algorithms.
        \item $\frac{\partial m^s_{k-1}}{\partial \theta}, \frac{\partial P^s_{k-1}}{\partial \theta}$:  naively, this would require applying the chain rule through all the steps of the convergence loop in Algorithm~\ref{alg:fixed_point_iteration}.
    \end{enumerate}
    Thankfully, $\mathfrak{n}^*_{0:n}$ verifies the fixed point identity $\mathfrak{n}^*_{0:n} = \mathrm{KS}(\mathfrak{n}^*_{0:n}, y_{1:n}, \theta)$, where $\mathrm{KS}$ is the Kalman smoother operator given by Algorithm~\ref{alg:fixed_point_iteration}. Therefore, we can leverage the implicit function theorem as in~\cite{christianson1994fixedpoint} to show that
    \begin{equation}
        \begin{split}
        \frac{\mathrm{d}m^s_{0:n}}{\mathrm{d}\theta} &= \frac{\partial \mathrm{KS}^{m^s}_{0:n}}{\partial m^s_{0:n}} \frac{\mathrm{d}m^s_{0:n}}{\mathrm{d}\theta} + \frac{\partial \mathrm{KS}^{m^s}_{0:n}}{\partial P^s_{0:n}} \frac{\mathrm{d}P^s_{0:n}}{\mathrm{d}\theta} + \frac{\partial \mathrm{KS}^{m^s}_{0:n}}{\partial \theta},\\
        \frac {\mathrm{d}P^s_{0:n}}{\mathrm{d}\theta}
        &= \frac{\partial \mathrm{KS}^{P^s}_{0:n}}{\partial P^s_{0:n}} \frac{\mathrm{d}P^s_{0:n}}{\mathrm{d}\theta} + \frac{\partial \mathrm{KS}^{P^s}_{0:n}}{\partial m^s_{0:n}} \frac{\mathrm{d}m^s_{0:n}}{\mathrm{d}\theta} + \frac{\partial \mathrm{KS}^{P^s}_{0:n}}{\partial \theta},
        \end{split}
    \label{eq:implicity_function_theorem}
    \end{equation}
    where $\mathrm{KS}^{m^s}_k$ and $\mathrm{KS}^{P^s}_k$ are the mean and covariance outputs of the smoother operator $\mathrm{KS}$ at time $k$, respectively. This ensures that, for any $k=0, \ldots, n$, $\frac{\mathrm{d}m^s_k}{\mathrm{d}\theta}$ and $\frac{\mathrm{d}P^s_k}{\mathrm{d}\theta}$ verify a linear fixed-point equation, which, provided that the operator norm of all $\frac{\partial \mathrm{KS}}{\partial m^s_k}$, $\frac{\partial \mathrm{KS}}{\partial P^s_k}$ is less than $1$, admits a unique solution which can be computed by iterating it.
    We can therefore replace the memory-unbounded reverse accumulation calculation of the gradient of the optimal nominal trajectory with a memory-constant iterative procedure. 
    The full procedure is summarized in Algorithm~\ref{alg:fixed_point_iteration_for_gradient}. 
    \begin{remark}
         It is worth noting that~\eqref{eq:implicity_function_theorem} and the related routine in Algorithm~\ref{alg:fixed_point_iteration_for_gradient} \textit{a priori} seem to have a $O(T^2)$ memory cost due to the apparent need of computing all the different combinations of derivatives. However, automatic differentiation libraries typically require manipulating vector-Jacobian (or Jacobian-vector) products rather than full Jacobians~\cite{jax2018github}, so that the actual implementation of Algorithm~\ref{alg:fixed_point_iteration_for_gradient} only incurs a memory cost of $O(T)$, in line with that of computing the log-likelihood in the first place.
    \end{remark}
   
    \begin{algorithm}[!t]
    \caption{Fixed-point algorithm for computing the gradient of the covariance-based nominal trajectory w.r.t. the state-space model parameters $\theta$ at convergence.}
    \label{alg:fixed_point_iteration_for_gradient}
        \begin{algorithmic}[1]
            \renewcommand{\algorithmicrequire}{\textbf{Input:}}
            \renewcommand{\algorithmicensure}{\textbf{Output:}}
            \REQUIRE Optimal trajectory $\mathfrak{n}^*_{0:n} = \{m^s_{0:n}, P^s_{0:n}\}$ computed with Algorithm~\ref{alg:fixed_point_iteration} for the parameter $\theta$.
            \ENSURE  $\left\{\frac{\mathrm{d}m^s_{k}}{\mathrm{d}\theta} , \, \frac {\mathrm{d}P^s_{k}}{\mathrm{d}\theta};\, k=0, \ldots, n\right\}$.
                \STATE Initialize $\frac{\mathrm{d}m^s_{k}}{\mathrm{d}\theta}, \, \frac {\mathrm{d}P^s_{k}}{\mathrm{d}\theta}$ for all $k=0, \ldots, n$
                \STATE Compute $\frac{\partial \mathrm{KS}^{m^s}_{0:n}}{\partial m^s_{0:n}}$, $\frac{\partial \mathrm{KS}^{m^s}_{0:n}}{\partial \theta}$, $\frac{\partial \mathrm{KS}^{P^s}_{0:n}}{\partial P^s_{0:n}}$, and $\frac{\partial \mathrm{KS}^{P^s}_{0:n}}{\partial \theta}$ at the optimal nominal trajectory $\mathfrak{n}^*_{0:n}$ and $\theta$
                \STATE Iterate~\eqref{eq:implicity_function_theorem} until convergence.
            \end{algorithmic}
    \end{algorithm}
    When $\mathfrak{n}^*$ has been computed, all the operations happening within the loop of Algorithm~\ref{alg:fixed_point_iteration_for_gradient} are fully parallelizable, apart from computing the terms $\frac{\partial \mathrm{KS}^{m^s}_{0:n}}{\partial m^s_{0:n}}$ and $\frac{\partial \mathrm{KS}^{m^s}_{0:n}}{\partial m^s_{0:n}}$, and the final sum of the incremental log-likelihood terms, which all have span complexity $O(\log(n))$ as they share the same computational graph as the one used to compute the smoothing solution itself. The total span complexity of computing the (pseudo) score function is, therefore, $O(\log(n))$.

\section{Experimental results}
\label{sec:experiment}

In this section, we empirically demonstrate the performance of our proposed framework. For this, we consider a population model~\cite{fasiolo2016comparison, tronarp2018iterative} and a coordinated turn (CT) model~\cite{bar2004estimation, roth2014ekf}. We aim to show the effectiveness of our methods in terms of computational complexity, robustness, and parameter estimation. All the experiments are conducted using the NVIDIA\textsuperscript{\textregistered} GeForce RTX 3080 Ti 12 GB with 10240 CUDA cores, and leveraging the Python library JAX~\cite{jax2018github} for parallelization and automatic differentiation support. The code to reproduce the results can be found at \url{https://github.com/EEA-sensors/sqrt-parallel-smoothers/}.

\subsection{Parallel inference in a general SSM}\label{subsec:speed-general}
We now evaluate the performance of the general parallel method for a fully nonlinear non-Gaussian SSM. This generalizes the experiment done in~\cite{yaghoobi2021parallel} which was only concerned with additive noise models. We consider the stochastic Ricker model, described in~\cite{fasiolo2016comparison}, and also used in ~\cite{tronarp2018iterative} as
\begin{equation}
    \begin{split}
        x_k &= \log(a) + x_{k-1} - \exp(x_{k-1}) + q_k, \quad q_k \sim N(0, Q_k)\\
        y_k \mid x_k &\sim \mathrm{Poisson}(b \exp(x_k)),  \\
        x_0 &\sim \delta(x_0 - \log(c)).
    \end{split}
\end{equation}
where $x_k$ is a one-dimensional state, and the measurement $y_k$ is conditionally distributed according to a Poisson distribution with rate $b \exp(x_k)$. For this experiment, the  system parameters are chosen to be $a=44.7$, $b=10$, $c=7$, and $Q_k=0.3^2$.

Our aim is to compare the performance of the covariance-based sequential and parallel methods in the sense of time span complexity. To this end, we perform $M=100$ iterations of the iterated extended (IEKS) and cubature Kalman smoothing (ICKS) on the model for different numbers of observations.

The average (over 100 repetitions) run times of sequential and parallel IEKS and ICKS methods are reported in Fig.~\ref{fig:gpu-poisson-model}. This figure shows that the parallel methods are generally faster compared to their sequential counterpart. However, the differences are less prominent for shorter time horizons (say, $n = 10$), and more significant with long time horizons. The run times of sequential IEKS and ICKS are increasing linearly with respect to the number of time steps, whereas the parallel IEKS and ICKS exhibit a much better complexity scaling. %
\subsection{Computational complexity of parallel square-root methods}\label{subsec:coordinate-complexity}
In this experiment, the goal is to show the effectiveness of our proposed square-root method in terms of computational complexity in a high-dimensional state-space setting. We consider the problem of tracking a maneuvering target with a CT model~\cite{bar2004estimation} given bearing-only sensor measurements. Other uses of this model can be found, for example, in~\cite{arasaratnam2009cubature} and~\cite{roth2014ekf}.
The state is a 5-dimensional vector $x = [p_x, p_y, \dot{p}_x, \dot{p}_y, \omega]^\top$ containing the position $(p_x, p_y)$, the speed $(\dot{p}_x, \dot{p}_y)$, and the turn rate $\omega$ of the target. The bearing is measured by two sensors located at known positions $(s_x^i, s_y^i)$ for $i=1,2$. We follow the parameterization of~\cite{arasaratnam2009cubature}, and take the noise parameters and time step to be $q_1 = 0.1$, $q_2 = 0.1$, and $dt = 0.01$.

Similarly as in Section~\ref{subsec:speed-general}, in this experiment, we consider observation sets with size varying from $10$ to $5\,000$. Fig.~\ref{fig:gpu-bearing-only} shows the average run time (computed over 100 runs) of running $M=10$ iterations of the iterated smoothing methods. As expected, the parallel square-root method is faster than its sequential square-root method. 

\begin{figure}[tbhp]
    \centering
    \subfloat[Population model\label{fig:gpu-poisson-model}.]{
      \resizebox{0.45\columnwidth}{!}{\begin{tikzpicture}
\begin{axis}[
    legend style={
        nodes={scale=0.75, transform shape},
        at={(0,1)},
        anchor=north west},
    grid=both,
    xmode=log, 
    ymode=log,
    xmin=10, xmax=130,
    ymin=0, ymax=5,
    extra x ticks={3584},
    extra x tick style={%
        grid=major,
    },
    extra x tick labels={n cores},
    xlabel={$n$ number of time steps},
    ylabel={run time (in seconds)}
    ]
\addplot[black, dashed, line width=1pt] table [x=time_steps, y=gpu_IEKS_std_seq_time, col sep=comma]{figures/gpu-population-model.csv};
\addplot[black, line width=1pt] table [x=time_steps, y=gpu_IEKS_std_par_time, col sep=comma]{figures/gpu-population-model.csv};
\addplot[gray, dashed, line width=1pt] table [x=time_steps, y=gpu_IPLS_std_seq_time, col sep=comma]{figures/gpu-population-model.csv};
\addplot[gray, line width=1pt] table [x=time_steps, y=gpu_IPLS_std_par_time, col sep=comma]{figures/gpu-population-model.csv};
\legend{sequential IEKS, parallel IEKS, sequential ICKS, parallel ICKS}
\end{axis}
\end{tikzpicture}}}
    \subfloat[Coordinated turn model\label{fig:gpu-bearing-only}.]{
        \resizebox{0.45\columnwidth}{!}{\begin{tikzpicture}
\begin{axis}[
    legend style={
        nodes={scale=0.75, transform shape},
        at={(0, 1)},
        anchor=north west},
    grid=both,
    xmode=log, 
    ymode=log,
    xmin=10, xmax=4096,
    ymin=0, ymax=2.5e2,
    xlabel={$n$ number of time steps},
    ylabel={run time (in seconds)},
    ]
\addplot[gray, dashed, line width=1pt] table [x=time_steps, y=gpu_extended_sqrt_seq_mean_time, col sep=comma]{figures/rtx-GPU-all-methods.csv};
\addplot[gray, line width=1pt] table [x=time_steps, y=gpu_extended_sqrt_par_mean_time, col sep=comma]{figures/rtx-GPU-all-methods.csv};
\addplot[black, dashed, line width=1pt] table [x=time_steps, y=gpu_cubature_sqrt_seq_mean_time, col sep=comma]{figures/rtx-GPU-all-methods.csv};
\addplot[black, line width=1pt] table [x=time_steps, y=gpu_cubature_sqrt_par_mean_time, col sep=comma]{figures/rtx-GPU-all-methods.csv};

\legend{sequential IEKS, parallel IEKS, sequential ICKS, parallel ICKS}

\end{axis}
\end{tikzpicture}}}
    \caption{GPU run time comparison of (a) the standard parallel and sequential methods of IEKS and ICKS methods for population model and (b) the square-root parallel and sequential methods of IEKS and ICKS methods for coordinated turn model.}
    \label{fig:GPU_std_sqrt}
\end{figure}
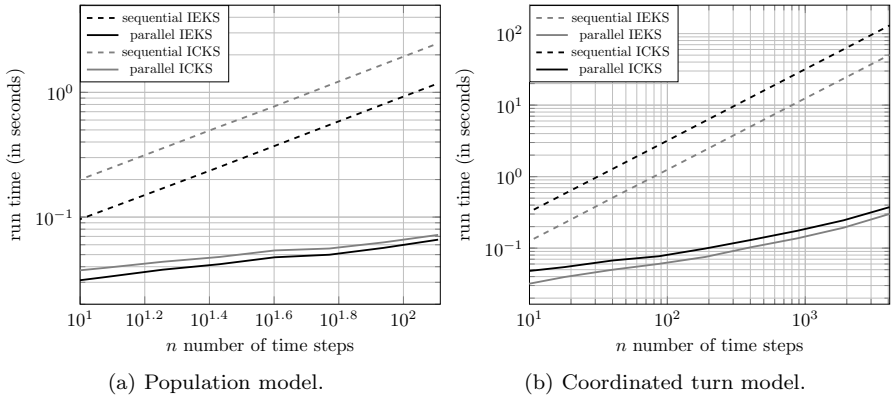

\subsection{Robustness of parallel square-root methods}
In this part, we study the same CT model as described in Section~\ref{subsec:coordinate-complexity} to show how the square-root method improves the stability of the algorithm, as compared to the covariance-based one, when using $32$-bit floating-point numbers. In order to assess the stability of the method, for each chosen number of observations, we simulate $15$ independent sets of observations and run square-root and covariance-based parallel iterated extended and cubature Kalman smoother for $M=20$ iterations. We then report the divergence rate for a given method and number of observations, as defined by the percentage of resulting NaN (not a number) log-likelihood estimates. The result is shown in Fig.~\ref{fig:ell}, and shows that square-root methods achieve higher stability in comparison to the standard method. 

\begin{figure}[tbhp]
    \centering
    \subfloat[Divergence rate of IEKS method\label{fig:ell-IEKS}.]{
      \resizebox{0.45\columnwidth}{!}{\begin{tikzpicture}
\begin{axis}[
    legend style={
        nodes={scale=0.75, transform shape},
        at={(0,1)},
        anchor=north west},
    grid=both,
    grid style=dashed,
    xmin=100, xmax=8000,
    xlabel={$n$ number of time steps},
    ylabel={divergence rate (\%)},
    ]
\addplot[black, dashed,mark=square, line width=1pt] table [x=observations, y=fl32_gpu_extended_std_par_ell, col sep=comma]{figures/rtx_fl32_extended_ell.csv};
\addplot[black,mark=square, line width=1pt] table [x=observations, y=fl32_gpu_extended_sqrt_par_ell, col sep=comma]{figures/rtx_fl32_extended_ell.csv};
\legend{std-IEKS, sqrt-IEKS}
\end{axis}
\end{tikzpicture}}}
    \subfloat[Divergence rate of ICKS method\label{fig:ell-ICKS}.]{
        \resizebox{0.45\columnwidth}{!}{\begin{tikzpicture}
\begin{axis}[
    legend style={
        nodes={scale=0.75, transform shape},
        at={(0,1)},
        anchor=north west},
    grid=both,
    grid style=dashed,
    xmin=1e2, xmax=8e3,
    xlabel={$n$ number of time steps},
    ylabel={divergence rate (\%)},
    ]
\addplot[black, dashed,mark=square, line width=1pt] table [x=observations, y=fl32_gpu_cubature_std_par_ell, col sep=comma]{figures/rtx_fl32_cubature_ell.csv};
\addplot[black,mark=square, line width=1pt] table [x=observations, y=fl32_gpu_cubature_sqrt_par_ell, col sep=comma]{figures/rtx_fl32_cubature_ell.csv};
\legend{std-ICKS, sqrt-ICKS}
\end{axis}
\end{tikzpicture}}}
    \caption{Divergence rate comparison, that is, the percentage of each time-step that becomes nan in the $15$ runs using 32-bit floating-point numbers, of the parallel square-root and standard versions of (a) IEKS and (b) ICKS methods.}
    \label{fig:ell}
\end{figure}
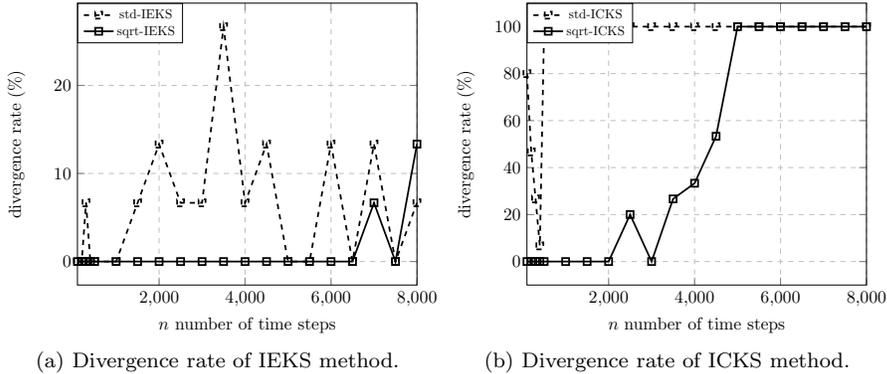
\subsection{Parallel parameter estimation}
The section demonstrates the  performance of the proposed parameter estimation algorithm in terms of accuracy and computational cost reduction using the CT model and a parallel method on GPU. Here, The parameter of interest is the standard deviation of the first sensor, $\sqrt{R_{1,1}}$, with a ground-truth value of $0.05$. The parameters used in this study are the same
parameters used by~\cite{kokkala2015sigma}. In order to assess the performance of the method, we simulate $10$ different trajectories and calibrate the model on varying size subsets of the generated sets of observations. To compute the log-likelihood of the model and the gradient thereof, we use Taylor series and several sigma-point linearization methods (cubature, unscented, and Gauss--Hermite) and the standard version of  iterated smoothers. The parameter is then estimated by maximum likelihood estimation, where the optimization is carried out using the L-BFGS-B algorithm with known (automatically computed) Jacobian. An important benefit of the proposed parallel method is its reduced computing time. To demonstrate this, we provide the GPU run time for various linearization schemes using both sequential and parallel methods. Figure~\ref{fig:time_sdR11} shows that the parallel methods take less time to run than the sequential methods across different linearization schemes.

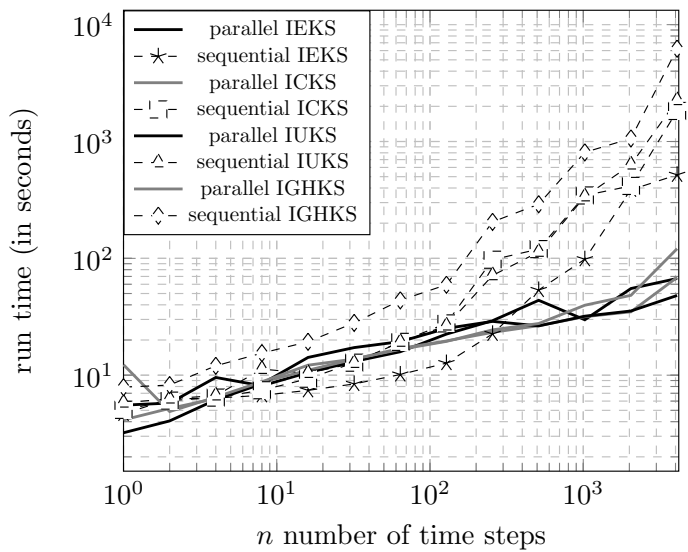
\begin{figure}[tbhp]
    \centering
        \subfloat{
      \resizebox{0.7\columnwidth}{!}{\begin{tikzpicture}
\begin{axis}[
    legend style={
        nodes={scale=0.75, transform shape},
        at={(0,1)},
        anchor=north west},
    grid=both,
    grid style=dashed,
    xmin=1, xmax=4200,
    xlabel={$n$ number of time steps},
    ylabel={run time (in seconds)},
    ymode=log,
    xmode=log
    ]
\addplot[black, line width=1pt] table [x=time_steps, y=time_extended_par, col sep=comma]{figures/GPU_sdR11_runtime.csv};
\addplot[black,dashed, mark=star, mark options={scale=1.5,fill=white}] table [x=time_steps, y=time_extended_seq, col sep=comma]{figures/GPU_sdR11_runtime.csv};
\addplot[gray, line width=1pt] table [x=time_steps, y=time_cubature_par, col sep=comma]{figures/GPU_sdR11_runtime.csv};
\addplot[black,dashed, mark=square*, mark options={scale=1.5,fill=white}] table [x=time_steps, y=time_cubature_seq, col sep=comma]{figures/GPU_sdR11_runtime.csv};
\addplot[black, line width=1pt] table [x=time_steps, y=time_unscented_par, col sep=comma]{figures/GPU_sdR11_runtime.csv};
\addplot[black,dashed, mark=triangle*, mark options={scale=1.5,fill=white}] table [x=time_steps, y=time_unscented_seq, col sep=comma]{figures/GPU_sdR11_runtime.csv};
\addplot[gray, line width=1pt] table [x=time_steps, y=time_gh_par, col sep=comma]{figures/GPU_sdR11_runtime.csv};
\addplot[black,dashed, mark=diamond*, mark options={scale=1.5,fill=white}] table [x=time_steps, y=time_gh_seq, col sep=comma]{figures/GPU_sdR11_runtime.csv};
\legend{parallel IEKS, sequential IEKS, parallel ICKS, sequential ICKS, parallel IUKS, sequential IUKS, parallel IGHKS, sequential IGHKS}
\end{axis}
\end{tikzpicture}}}
    \caption{GPU run time comparison of $\sqrt{R_{1,1}}$ estimate using standard parallel/sequential methods of the iterated extended (IEKS), cubature (ICKS), unscented (IUKS), and Gauss-Hermite (IGHKS) Kalman smoothers.}
    \label{fig:time_sdR11}
\end{figure}

\section{Conclusions}
\label{sec:conclusions}
In this paper, we have presented a novel parallel-in-time inference method for general SSMs. First, to improve the numerical stability of the parallel-in-time Kalman filter and RTS smoother algorithms, we developed square-root versions of the method. Then we considered general SSMs and leverage GSLR. By reformulating the GSLR equations in terms of the square root of covariance matrices, we then obtained a square-root version of our parallel-in-time nonlinear smoother. Finally, using the fixed-point structure of the iterated method, we showed how to compute the log-likelihood and gradient thereof in parallel for use in parameter estimation. Experimental results conducted on a GPU showed that the proposed inference method reduces time complexity and improves numerical stability. 

\section*{Authors contributions}
Original idea by S.S. and S.H.; derivations by F.Y. and S.S.; implementation by A.C. and F.Y.; parameter estimation by A.C.; writing primarily by F.Y.

\bibliographystyle{siamplain}
\bibliography{references}
\end{document}